%% file: main.tex
\documentclass[sigconf, nonacm=true]{acmart}
\usepackage{graphicx}
\usepackage{balance}  

\usepackage[ruled,vlined,linesnumbered]{algorithm2e}
\usepackage{booktabs}
\usepackage{multirow}
\usepackage{multicol}
\usepackage{xcolor}
\usepackage{verbatim}
\usepackage{balance}

\usepackage{amssymb,amsmath}
\newtheorem{theorem}{Theorem}[section]
\newtheorem{lemma}[theorem]{Lemma}
\newtheorem{example}[theorem]{Example}

\newtheorem{definition}[theorem]{Definition}
\newtheorem{corollary}[theorem]{Corollary}

\newenvironment{sproof}{%
  \renewcommand{\proofname}{Proof sketch}\proof}{\endproof}

\newcommand*{\q}[1]{\textcolor{red}{#1}}
\newcommand*{\ye}[1]{\textcolor{purple}{#1}}
\newcommand*{\rv}[1]{\textcolor{black}{#1}}

\AtBeginDocument{%
  \providecommand\BibTeX{{%
    \normalfont B\kern-0.5em{\scshape i\kern-0.25em b}\kern-0.8em\TeX}}}

\copyrightyear{2021}
\acmYear{2021}
\setcopyright{acmcopyright}
\acmConference[SIGMOD '21] {Proceedings of the 2021 International Conference on Management of Data}{June 20--25, 2021}{Virtual Event, China}
\acmBooktitle{Proceedings of the 2021 International Conference on Management of Data (SIGMOD '21), June 20--25, 2021, Virtual Event, China}

\acmPrice{15.00}
\acmDOI{10.1145/3448016.3452826}
\acmISBN{978-1-4503-8343-1/21/06}

\begin{document}
\fancyhead{}

\title{Query-by-Sketch: Scaling \rv{Shortest Path Graph} Queries \\on Very Large Networks}


\author{Ye Wang}
\affiliation{School of Computing, Australian National University}
\email{ye.wang2@anu.edu.au}
\author{Qing Wang}
\affiliation{School of Computing, Australian National University}
\email{qing.wang@anu.edu.au}
\author{Henning Koehler}
\affiliation{School of Fundamental Sciences, Massey University}
\email{h.koehler@massey.ac.nz}
\author{Yu Lin}
\affiliation{School of Computing, Australian National University}
\email{yu.lin@anu.edu.au}

\renewcommand{\shortauthors}{}
\renewcommand{\shorttitle}{}

\begin{abstract}
Computing shortest paths is a fundamental operation in processing graph data. In many real-world applications, discovering shortest paths between two vertices empowers us to make full use of the underlying structure to understand how vertices are related in a graph, e.g. the strength of social ties between
individuals in a social network. In this paper, we study the shortest-path-graph problem that aims to efficiently compute a shortest path graph containing exactly all shortest paths between any arbitrary pair of vertices on complex networks. Our goal is to design an exact solution that can scale to graphs with millions or billions of vertices and edges. To achieve high scalability, we propose a novel method, \emph{Query-by-Sketch} (QbS), which efficiently leverages offline labelling (i.e., precomputed labels) to guide online searching through a fast sketching process that summarizes the important structural aspects of shortest paths in answering shortest-path-graph queries. We theoretically prove the correctness of this method and analyze its computational complexity. To empirically verify the efficiency of QbS, we conduct experiments on 12 real-world datasets, among which the largest dataset has 1.7 billion vertices and 7.8 billion edges. The experimental results show that QbS can answer shortest-path-graph queries in microseconds for million-scale graphs and less than half a second for billion-scale graphs. 
 
\end{abstract}



\maketitle

\input{section_introduction.tex}

\input{section_preliminaries.tex}

\input{section_labelling-based.tex}

\input{section_query_by_sketch.tex}
\input{subsection_label.tex}
\input{subsection_sketch.tex}

\input{subsection_search.tex}

\input{section_theoretical-analysis.tex}

\input{section_experiment.tex}

\input{section_related-work.tex}

\section{Conclusions}\label{sec:conclusion}
We have proposed a novel method QbS to answer shortest-path-graph queries on large graphs. QbS constructs a labelling scheme through pre-computation, and then answers queries by performing online computation that involves fast sketching and guided searching. We have analyzed the complexity and correctness of our method. Our labelling scheme is deterministic and can be constructed through a parallelized process. We have conducted experiments on 12 large real-world graphs to empirically verify the scalability and efficiency of QbS. 
For future work, we plan to extend QbS on road networks by leveraging their specific properties and study landmark selection strategies to improve the performance.

\clearpage
\bibliographystyle{ACM-Reference-Format}
\balance
\bibliography{ref}

\clearpage
\appendix

\end{document}

%% file: section_introduction.tex
\section{Introduction}

\label{sec:introduction}

Graphs are typical data structures used for representing complex \rv{relationships} among entities, such as friendships in social networks, connections in computer networks, and links among web pages \cite{scott1988social,boccaletti2006complex,ukkonen2008searching}. Computing shortest paths between vertices is a fundamental operation in processing graph data, and has been used in many algorithms for graph analytics \cite{yao2013secure,opsahl2010node,kolaczyk2009group}. These algorithms are often applied to support applications that require low latency on graphs with millions or billions of vertices and edges. Therefore, it is highly desirable -- but challenging -- to compute shortest paths efficiently on very large graphs. 

Previously, the problem of point-to-point shortest path queries has been well studied, which is to find a shortest path between two vertices in a graph \cite{goldberg2005computing,goldberg2006reach,bast2007transit,goldberg2007point,wagner2007speed,abraham2010highway,wu2012shortest,sankaranarayanan2009path,sanders2005highway}. By leveraging specific properties of road networks, such as hierarchical structures and near planarity \cite{fu2013label,akiba2013fast}, previous works have proposed various exact and approximate methods for answering point-to-point shortest path queries \cite{cowen2004compact,abraham2012hierarchical}. Nonetheless, these methods often do not perform well on complex networks (e.g., social networks, and web graphs) because complex networks exhibit different properties from road networks, such as small diameter and local clustering \cite{goldberg2005computing,fu2013label,akiba2013fast}. Furthermore, existing methods for point-to-point shortest path queries were designed with the guarantee of finding only one shortest path, which limits their usability in practical applications.

\begin{figure}[ht!]
    \centering
    \includegraphics[width=5.5cm]{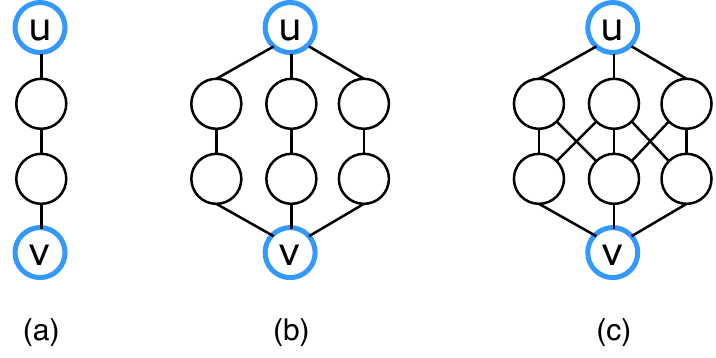}\vspace{-0.2cm}
    \caption{An illustration of shortest paths between two vertices $u$ and $v$ whose distance is $3$: (a) one shortest path; (b) three shortest paths; (c) \rv{seven shortest paths.} }
    \label{fig:int}\vspace{-0.3cm}
\end{figure}

\begin{figure*}[ht!]
\centering
\includegraphics[width=17cm]{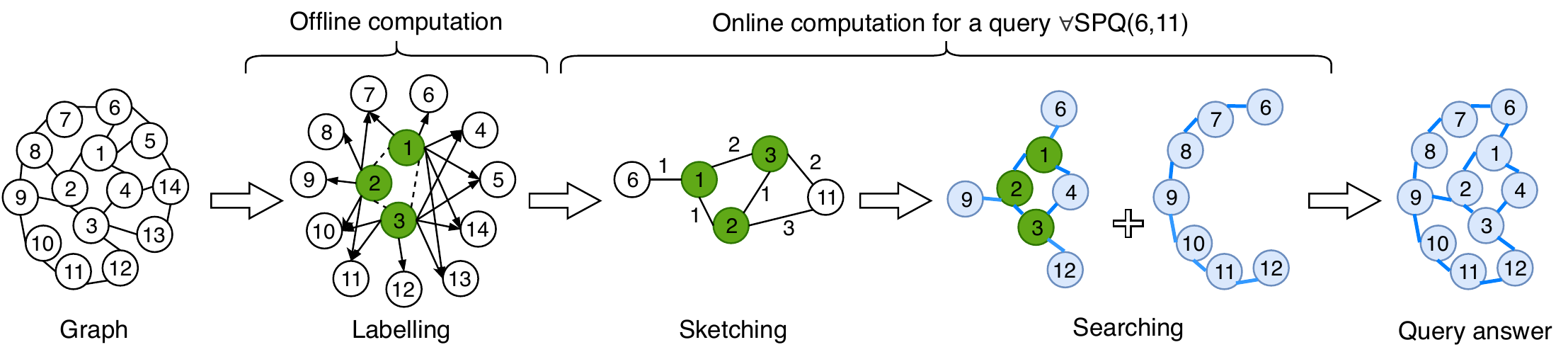}\vspace{-0.2cm}
\caption{An illustration of our method Query-by-Sketch (QbS) for answering all shortest path queries.}\vspace{-0.3cm}
\label{fig:frame}
\end{figure*}
\rv{Given two vertices $u$ and $v$, as depicted in Figure \ref{fig:int}(a)-(c), they have the same distance and cannot be distinguished from one another if only one shortest path is considered. However, when considering all shortest paths, the shortest paths between these two vertices indeed exhibit considerably different structures in Figure \ref{fig:int}(a)-(c), which can not only distinguish vertices $u$ and $v$ in different scenarios, but also empower us to make full use of such structures to analyze how they are connected. }\rv{Thus, in this paper, we study the problem of finding the structure of shortest paths between vertices. Specifically, we use the notion of ``shortest path graph" to represent the structure of shortest paths between two vertices, which is a subgraph containing \emph{exactly all shortest paths} between these two vertices. Accordingly, we term this problem as the \emph{shortest-path-graph} problem (formally defined in Section \ref{sec:prelimiaries}).}

\rv{Interestingly, shortest path graph manifests itself as a basis for tackling various shortest path related problems, particularly when investigating the structure of the solution space of a combinatorial problem based on shortest paths, 
for example, the Shortest Path Rerouting problem (i.e., to find a rerouting sequence from one shortest path to another shortest path that only differs in one vertex) \cite{kaminski2011shortest,bonsma2013complexity,nishimura2018introduction}, the Shortest Path Network Interdiction problem (i.e., to find critical edges and vertices whose removal can destroy all shortest paths between two vertices) \cite{khachiyan2008short,israeli2002shortest}, and the variants such as the Shortest Path Common Links problem (i.e., to find links common to all shortest paths between two vertices) \cite{labbe1995location,hansen1986efficient}. These shortest path related problems are motivated by a wide range of real-world applications arising in designing and analyzing networks. For example, identifying a rerouting sequence for shortest paths enables the robust design of networks with minimal cost for reconfiguration, and finding critical edges and vertices helps defend critical infrastructures against cyberattacks.}

\rv{However, computing shortest path graphs is computationally expensive since it requires to identify all shortest paths, not just one, between two vertices.} A straightforward solution for answering \rv{shortest-path-graph queries} is to compute on-the-fly all shortest paths between two vertices using Dijkstra algorithm for weighted graphs {\cite{dijkstra1959note}} or performing a breadth-first search (BFS) for unweighted graphs \cite{cormen2009introduction}. 
This is costly on graphs with millions or billions of vertices and edges.
\rv{Another solution is to precompute all shortest paths for all pairs of vertices in a graph and then assign precomputed labels to vertices such that certain properties hold, e.g. 2-hop distance cover \cite{cohen2003reachability}.} However, for large graphs, storing even just shortest path distances of all pairs of vertices is prohibitive \cite{akiba2013fast} and storing all shortest paths of all pairs is hardly feasible due to the demand for much more space overhead.  
Thus, the question we tackle in this paper is: \emph{How to construct labels for \rv{shortest-path-graph queries} that should be of reasonable size (e.g. not much larger than the original graph), within a reasonable time (e.g. not longer than one day), and can speed up query answering as much as possible?}
In answering this question, we develop an efficient solution for \rv{ shortest-path-graph queries.} It is worth to note that: \rv{1) we do not enumerate all shortest paths to produce a shortest path graph that contains \emph{exactly all shortest paths} between two vertices; 2) our proposed solution can answer \rv{shortest-path-graph queries} very efficiently, in microseconds for graphs with millions of edges and in less than half a second for graphs with billions of edges.} 

 \vspace{0.2cm} 
\noindent\textbf{Contributions.~} In the following, we summarize the contributions of this paper with the key technical details:

 \vspace{0.1cm} 
   \noindent(1) We observe that 2-hop distance cover is inadequate for labelling required by \rv{shortest-path-graph} queries. 
    To alleviate this limitation and achieve high scalability, we propose a scalable method for answering \rv{shortest-path-graph} queries, called \emph{Query-by-Sketch} (QbS).
    This method consists of three phases, as illustrated in Figure \ref{fig:frame}: \rv{(a)} \emph{labelling} - constructing a labelling scheme, which is compact and of a small size, using a small number of landmarks through precomputation, \rv{(b)} \emph{sketching} - using labelling to efficiently compute a \emph{sketch} that summarizes the important structure of shortest paths in a query answer, and (c) \emph{searching} - \rv{computing shortest paths on a sparsified graph} under the "guide" of the sketch. We develop efficient algorithms for these phases, and combine them effectively \rv{to handle shortest-path-graph queries} on very large graphs.

\smallskip    
   \noindent(2)     
    We theoretically prove the correctness of our method $QbS$. In addition to this, we conduct the complexity analysis for $QbS$ through analysing the time complexities of the algorithms for constructing a labelling scheme, computing a sketch, and performing a guided search for answering queries. We also prove that our labelling scheme is deterministic w.r.t. landmarks. This enables us to leverage the thread-level parallelism by performing BFSs from different landmarks simultaneously without considering an order of landmarks, which improves the efficiency of labelling construction and thus achieves better scalability.

\smallskip    
   \noindent(3)     
    We have conducted experiments on 12 real-world datasets, among which the largest dataset ClueWeb09 has 1.7 billion vertices and 7.8 billion edges. It is shown that $QbS$ has significantly better scalability than the baseline methods. The labelling construction of $QbS$ can be  
   parallelized, which takes 10 seconds for datasets with millions of edges and half an hour for the largest dataset ClueWeb09. The labelling sizes constructed by $QbS$ are generally smaller than the original sizes of graphs. Further, $QbS$ can answer queries much faster than the other methods. For graphs with billions of edges, it takes only around 0.01 - 0.5 seconds to answer a query.


%% file: section_preliminaries.tex
\section{Preliminaries}\label{sec:prelimiaries}

Let $G=(V,E)$ be an unweighted graph, where $V$ and $E$ represent the set of vertices and edges in $G$, respectively. 
Without loss of generality, we assume that $G$ is undirected and connected since our work can be easily extended to directed or 
disconnected graphs. We use $V(G)$ and $E(G)$ to refer to the set of vertices and edges in $G$, respectively, $P_{uv}$ the set of all shortest paths between $u$ and $v$, and $d_G(u,v)$ the shortest path distance between $u$ and $v$ in $G$.

\medskip
\noindent\textbf{Distance labelling. }Let $R\subseteq V$ be a subset of special vertices in $G$, called \textit{landmarks}.
For each vertex $v\in V$, the \emph{label} of $v$ is a set of \emph{labelling entries} $L(v)=\{(r_1,\delta_{vr_1}), \dots, $ $(r_n,\delta_{vr_n})\}$, where $r_i\in R$ and $\delta_{vr_i}=d_G(v,r_i)$. We call $L=\{L(v)\}_{v\in V}$ a \textit{labelling} over $G$. The \emph{size} of a labelling $L$ is defined as \textit{size}(L)=$\Sigma_{v\in V}|L(v)|$. In viewing that each labelling entry $(r_i,\delta_{vr_i})$ corresponds to a \emph{hop} from a vertex $v$ to a landmark $r_i$ with the distance  $\delta_{vr_i}$, Cohen \emph{et al}. \cite{cohen2003reachability} proposed \emph{2-hop distance cover}, which has been widely used in labelling-based approaches for distance queries.
\begin{definition}\label{def:2-hop-distance-cover}{\emph{[\textsc{2-hop distance cover}]}}
A labelling $L$ over a graph $G=(V,E)$ is a \emph{2-hop distance cover} iff, for any two vertices $u,v\in V$, the following holds: 
\begin{displaymath}
d_G(u,v)=min\{\delta_{ur}+\delta_{vr}|(r,\delta_{ur})\in L(u),(r,\delta_{vr})\in L(v)\}.
\end{displaymath}
\end{definition}

\rv{Informally, 2-hop distance cover requires that, for any two vertices in a graph, their labels must contain at least one common landmark $r$ that lies on one of their shortest paths.}

\medskip
\noindent\textbf{\rv{Shortest-path-graph problem. }}\rv{In this work, we study shortest-path-graph queries. We first define the notion of \emph{shortest path graph}.}

\begin{definition}{\emph{[\textsc{Shortest path graph}]}}
Given any two vertices $u$ and $v$ in a graph $G$, the \textit{shortest~path~graph} (SPG) between $u$ and $v$ is a subgraph $G_{uv}$ of $G$, where (1) $V(G_{uv})=\bigcup_{p\in P_{uv}} V(p)$ and (2) $E(G_{uv})=\bigcup_{p\in P_{uv}} E(p)$.
\end{definition}

A shortest path graph $G_{uv}$ is different from an induced subgraph $G[V']$ where $V'=\bigcup_{p\in P_{uv}} V(p)$. Every edge in $G_{uv}$ must lie on at least one shortest path between $u$ and $v$, whereas $G[V']$ may contain edges that do not lie on any shortest path between $u$ and $v$. 

\begin{definition}{\emph{[\textsc{\rv{Shortest-path-graph problem}}]}} 
Let $G=(V,E)$ and $u,v\in V$. Then the \rv{\emph{shortest-path-graph problem}} is, given a query \rv{$SPG(u,v)$}, to find the shortest path graph $G_{uv}$ over $G$.
\end{definition}

%% file: section_labelling-based.tex
\section{Shortest Path Labelling}\label{sec:labelling-based-methods}

In this section, we discuss several labelling-based methods for the shortest-path-graph problem. The purpose is to discuss  their limitations and possible sources of difficulties.  

\subsection{2-Hop Path Cover}
Originally, 2-hop distance cover was proposed for reachability and distance queries \cite{cohen2003reachability}. 
Below, we discuss why it is insufficient for shortest-path-graph queries.

\begin{example}
Consider a query \rv{$SPG(3,7)$} on a graph $G$ depicted in Figure \ref{fig:new_ex} (a). The query answer is colored in green. In Figure \ref{fig:new_ex}(b), labels of a 2-hop distance cover over $G$ are colored in black. \rv{Starting from vertices $3$ and $7$, we can find vertex $1$ because $(1,1)\in L(3)$ and $(1,3)\in L(7)$, $d_G(3,7)=1+3=4$. Then, we have to stop since the label of vertex $1$ does not contain entries to other vertices. Thus, using the labels of the 2-hop distance cover can compute only one shortest path between $3$ and $7$, failing to find vertices $2$, $4$ and $5$ in the answer.} 
\end{example}

\rv{Finding a shortest path graph that exactly contains all shortest paths between two vertices} requires us to accurately encode \emph{every} shortest path between two vertices into labels. 
Thus, to answer \rv{shortest-path-graph} queries, we generalize 2-hop distance cover to a property called \emph{2-hop path cover}.

\begin{definition}\label{def:2-hop-path-cover}{\emph{[\textsc{2-hop path cover}]}}
Let $G=(V,E)$ be a graph and $L$ a labelling over $G$. We say $L$ is a \emph{2-hop path cover} iff $L$ is a 2-hop distance cover and, for any two vertices $u,v\in V$ and any path $p\in P^G_{uv}$ with $p\neq (u,v)$, the following holds:
\begin{align}\label{equ:2hop-path-cover}
\begin{split}
d_G(u,v)=min\{\delta_{ur}+\delta_{vr}| (r,\delta_{ur})\in L(u),\\  (r,\delta_{vr})\in L(v), r\in V(p)\backslash\{u,v\}\},
\end{split}
\end{align}

\end{definition}
Compared with 2-hop distance cover, 2-hop path cover further requires that, for any shortest path $p$ between any two vertices $u$ and $v$ that contains more than one edge, the labels of $u$ and $v$ should contain a common landmark $r$ that lies on $p$, but not be $u$ or $v$.



\begin{figure}
    \centering
    \includegraphics[width=8cm]{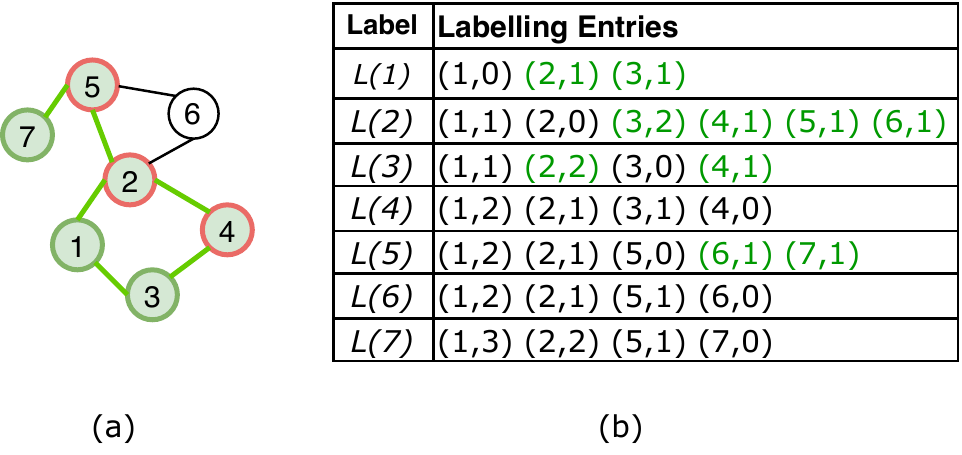}\vspace{-0.3cm}
    \caption{(a) A graph $G$ in which the answer of $\forall SPQ(3,7)$ is colored in green; (b) Labels over $G$, where labels for a 2-hop distance cover are colored in black and additional labels from a 2-hop path cover are colored in green.
    }
    \label{fig:new_ex}\vspace{-0.3cm}
\end{figure}

\begin{example}

Consider Figure \ref{fig:new_ex} again, in which a 2-hop path cover contains labels colored both in black and in green. According to the labels of vertices $1$ and $7$, vertex $2$ can be found. Then by the labels of $2$ and $7$, we can further find vertex $5$. Similarly, vertex $4$ can be found through the labels of $2$ and $3$. Thus, using the labels of the 2-hop path cover, we can find the query answer for \rv{$SPG(3,7)$}. 
\end{example}

\subsection{Path Labelling Methods}

To answer \rv{shortest-path-graph} queries, a naive labelling-based method is, for each vertex $v\in V$, to conduct a breadth-first search (BFS) from $v$ and store the distances between $v$ and all other vertices in the label of $v$, i.e. $L(v)=\{(u,\delta _{vu})|u\in V\}$, which is a 2-hop path labelling. 
Although \rv{shortest-path-graph} queries can be answered using $L$, it is inefficient, particularly when a graph is large. The time and space complexity of constructing such labels are $O(|V||E|)$ and $O(|V^2|)$ respectively. Answering one \rv{shortest-path-graph} query would cost $O(|V^2|)$ in the worst case. A question that naturally arises is: can we follow the idea of Pruned Landmark Labelling (PLL) \cite{akiba2013fast}, which has been shown to be successful for distance queries, to develop a pruning strategy for \rv{shortest-path-graph} queries for 
improving efficiency? We will thus introduce two pruned path labelling methods for \rv{shortest-path-graph} queries in the following.

\medskip
\noindent\textbf{Pruned path labelling.}
Inspired by Pruned Landmark Labelling (PLL) \cite{akiba2013fast}, we conduct pruning during the breadth-first searches, i.e. \emph{pruned} BFSs, for \rv{shortest-path-graph} queries. We abbreviate this pruned path labelling method by PPL.

PPL works as follows. Given a pre-defined landmark order $[v_1,v_2,$ $\dots,v_{|V|}]$ over all vertices in $G$, we conduct a pruned BFS from each vertex one by one as described in Algorithm \ref{alg:pruned-bfs}. 
In each pruned BFS rooted at $v_k$, we use $depth[v]$ to denote the distance between $v_k$ and $v$. \rv{Further, $L_{k-1}$ refers to the labels that have been constructed through the previous pruned BFSs from vertices $[v_1,\dots,v_{k-1}]$, and $d_{L_{k-1}}(v_k, u)$ denotes the distance between $v_k$ and $u$ being queried using labels in $L_{k-1}$. When $d_{L_{k-1}}(v_k, u)<depth[u]$, the label $(v_k, depth[u])$ is pruned (Lines 6-7) because labels in $L_{k-1}$ have already covered the shortest paths between $v_k$ and $u$. In other words, $v_k$ is only added into the labels of vertices $u$ when $d_{L_{k-1}}(v_k, u)\geq depth[u]$ (Line 8). Note that, unlike PLL, in the case of $d_{L_{k-1}}(v_k, u)$ $=depth[u]$, the label $(v_k, depth[u])$ cannot be pruned in PPL; otherwise, 2-hop path cover is not guaranteed, i.e., not all shortest paths are covered by labels. When $d_{L_{k-1}}(v_k, u)\leq depth[u]$, no further edges are traversed from $u$ because paths in this expansion have already been covered by labels in $L_{k}$ (Lines 6-7 and 9-10).}

\begin{algorithm}[ht!] 
\caption{PrunedBFS} \label{alg:pruned-bfs}
\KwIn{$G=(V,E)$; a landmark $v_k$; a labelling $L_{k-1}$}
    $Q\leftarrow \emptyset$;    $Q.push(v_k)$;\\
     $depth[v_k]\leftarrow 0$, $depth[v]\leftarrow \infty$ for all $v\in V\backslash\{v_k\}$;\\
    $L_{k}(v)\leftarrow L_{k-1}(v)$ for all $v\in V$;\\
    \While{$Q$ is not empty}
    {
        dequeue $u$ from $Q$;\\
        \If{$d_{L_{k-1}}(v_k,u)<depth[u]$}
        {
            continue;\\
        }
        $L_{k}(u)\leftarrow L_{k}(u)\cup \{(v_k,depth[u])\}$;\\
        \If{$d_{L_{k-1}}(v_k,u)=depth[u]$}
        {
            continue;
        }
        \For{all $(u,v_i)\in E$ s.t. $depth[v_i]=\infty$}
        {
            $depth[v_i]\leftarrow depth[u]+1$;\\
            enqueue $v_i$ to $Q$;\\
        }
    }
    \textbf{return} $L_k$;
\end{algorithm}

To answer a query \rv{$SPG(u,v)$}, we need to compute vertices and edges of $G_{uv}$ from a pruned path labelling $L$ recursively.
Assume that $d_{G}(u,v)\neq 1$; otherwise we finish with $G_{uv}$ containing only one edge $(u,v)$. 
We begin with $E(G_{uv})=\emptyset$. We find the common landmarks in their labels that are on the shortest paths, e.g., computing a set $V_{uv}=\{r|r=min(\delta _{ur}+\delta _{vr}), (r,\delta _{ur})\in L(u),(r,\delta _{vr})\in L(v)\}$. Then we query the shortest paths between u, v and these common landmarks, i.e., $(u,r)$ and $(v,r)$ for each $r\in V_{uv}$. The query $\forall SPQ(u,v)$ is computed by combining the shortest paths between u, v and the landmarks, i.e., $E(G_{uv})=\bigcup_{r\in V_{uv}}(E(G_{ur})\cup E(G_{vr}))$. 

\begin{example}\label{exa:ppl}
When using PPL to answer the query $\forall SPQ(3,7)$ on the graph $G$ in \rv{Figure} \ref{fig:new_ex}(a), we start with $(3,7)$ and obtain $V_{3,7}=\{1,2\}$. This leads to four new queries $(3,1), (7,1), (3,2)$ and $(7,2)$. The distance between $3$ and $1$ is 1. Thus, $E(G_{3,7})=\{(1,3)\}\cup E(G_{7,1}) \cup E(G_{3,2}) \cup E(G_{7,2})$. For the new query $(7,1)$, we obtain $V_{7,1}=\{2\}$, leading to another queries $(7,2)$ and $(1,2)$. Similarly, for $(3,2)$ and $(7,2)$ we obtain queries $(1,2)$, $(2,3)$, $(2,5)$ and $(2,7)$. Note that the labels of vertex $3$ are visited more than once, i.e. when querying $(3,7)$ and $(3,2)$. Further, because $3$ and $7$ have multiple shortest paths between them, more than one common vertex on their shortest paths \rv{are} found from their labels, i.e. $\{1,2\}$. As a result, edges $(2,5)$ and $(5,7)$ are handled multiple times, i.e., when querying $(2,7)$ and $(1,7)$.  

\end{example}

PPL has the same time and space complexity for constructing labels as the naive labelling-based method. However, due to pruning in BFSs, PPL can construct labels more efficiently with a significantly reduced labelling size. Nonetheless, the query time of PPL is still slow because all shortest paths between two vertices can only be found through searching vertices and edges using labels in a recursive manner. When more than one shortest path exists between query vertices, labels of some vertices are searched repeatedly and edges are found repeatedly, leading to unnecessary computational cost, e.g., vertex $3$ and edges $\{(2,5)(5,7)\}$ as in Example \ref{exa:ppl}.

\medskip
\noindent\textbf{Path labelling with parents.}  One common technique to accelerate query time for shortest-path-graph queries is to keep additional parent information in labels \rv{so as to provide a clearer direction towards shortest paths.} For example, 
Akiba \emph{et al}. \cite{akiba2013fast} extended the label of each vertex $v\in V$ to a set of triples $(r,\delta _{vr}, w_{vr})$ where $w_{vr}$ is the ``parent" vertex of $r$ on a shortest path from $v$ to $r$.
To find all shortest paths, this requires us to store all parent vertices of a vertex, rather than just one parent vertex as in the previous work for finding one shortest path. To be precise, we store a set of triples $\{(r_i,\delta _{vr_i}, W_{vr_i})\}_{1\leq i\leq |V|}$ where $W_{vr_i}$ is a set of ``parent" vertices of $v$ on a shortest path from $v$ to a landmark $r_i$. To reduce space overhead, for each of such shortest paths, we store the ``parent" vertices of $v$, rather than the ``child" vertices of $r_i$, because landmarks often have a high degree \cite{akiba2013fast}. 
To distinguish from PPL, we abbreviate this method with additional parent information by ParentPPL.

The time complexity of ParentPPL for constructing labels remains to be $O(|V||E|)$ but the space complexity becomes $O(|V||E|)$. In practice, additional parent information only helps speed up query time on small graphs. Even for a graph with millions of vertices and edges, \rv{ParentPPL would run out of time (same as PPL) or space,} failing to construct labels. We will discuss this further in Section \ref{sec:experiments}.

\subsection{Discussion}

For 2-hop labelling-based methods such as PPL and ParentPPL, the structure (i.e. shortest paths) of a graph is encoded into distance information of labels under the guarantee of 2-hop path cover. Although shortest paths can be recovered through computing distances between pairs of vertices, these methods are inefficient. \rv{This is because they recursively split each path into two sub-paths and compute vertices on sub-paths via distance information in labels, which leads to redundant or unnecessary searches.} 
\rv{Although storing parent information can often accelerate query time, it makes labelling size larger and does not scale over large networks.} Therefore, we need to find a method for which (1) the labelling size is small, (2) the structure of shortest paths can be recovered in an efficient way, i.e., reducing redundant and unnecessary computation, and (3) it can scale over large networks.

%% file: section_query_by_sketch.tex
\section{Query-by-Sketch}\label{sec:qbs}

In this section, we present an efficient and scalable method for solving the \rv{shortest-path-graph} problem, called \emph{Query-by-Sketch} (QbS). Conceptually, this method consists of three key components: \emph{labelling}, \emph{sketching} and \emph{searching}, which will be discussed in Sections \ref{subsec:labelling}, \ref{subsec:sketch} and \ref{subsec:searching}, respectively. The main idea behind this method is to construct a labelling scheme through precomputation, and then answer \rv{shortest-path-graph} queries by performing online computation that involves two steps: fast sketching and guided searching.

%% file: subsection_label.tex
\subsection{Labelling Scheme}
\label{subsec:labelling}

Let $G=(V,E)$ be a graph, $R\subseteq V$ be a set of landmarks, and $|R|<\!\!<|V|$ (i.e., $|R|$ is sufficiently smaller than $|V|$). 
We first preprocess the graph $G$ to obtain a compact representation of the shortest paths among landmarks, called a \emph{meta-graph} of $G$. 
Then, based on such a meta-graph, we define a labelling scheme to assign a label to each vertex in $G$ such that, given any pair of vertices $u,v\in V$, we can efficiently compute a sketch for answering $SPG(u,v)$.

\begin{definition}{\emph{[\textsc{Meta-graph}]}}\label{def:metaG}
A \emph{meta-graph} is $M=(R, E_R, \sigma)$ where $R$ is a set of landmarks, $E_R\subseteq R \times R$ is a set of edges s.t. $(r,r')\in E_R$ iff at least one shortest path between $r$ and $r'$ does not go through any other \rv{landmarks,} and $\sigma: E_R\mapsto \mathbb{N}$ assigns each edge in $E_R$ a weight, i.e. $\sigma(r,r')=d_G(r, r')$.
\end{definition}

Conceptually, a meta-graph represents how landmarks are connected through their shortest paths in a graph $G$. 

\begin{figure}
    \centering
    \includegraphics[width=7.5cm]{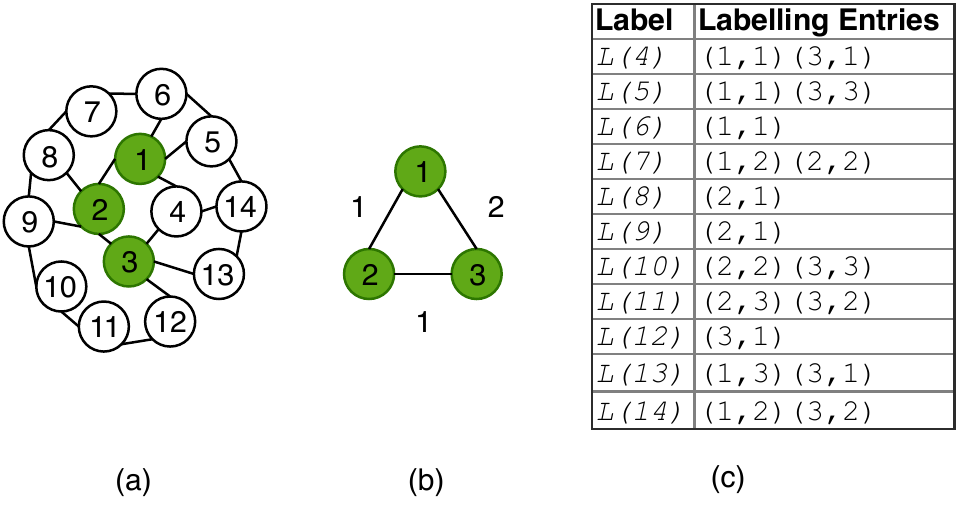}\vspace{-0.3cm}
    \caption{(a) A graph with three landmarks $\{1,2,3\}$ (highlighted in green), (b) a meta-graph, and  (c) a path labelling. }
    \label{fig:full-scheme}\vspace{-0.3cm}
\end{figure}

\begin{definition}{\emph{[\textsc{Labelling scheme}]}}\label{def:scheme}
 A \emph{labelling scheme} $\mathcal{L}=(M, L)$ consists of a meta-graph $M$ and a path labelling $L$ that assigns to each vertex $u\in V\backslash R$ a label $L(u)$ s.t. 
\begin{align}\label{equ:path-labelling}
\begin{split}
L(u)=\{(r, \delta_{ur})| r\in R,  \delta_{ur}=d_G(u,r), \\\exists p\in P_{ur} (V(p)\cap R=\{r\})
\}. 
\end{split}
\end{align}
\end{definition}
Note that, to accurately present how vertices are linked to landmarks, we only allow that $(r, \delta_{ur})$ is in the label $L(u)$ iff there exists at least one shortest path between $u$ and $r$ that does not contain other landmarks.
\ye{}
\begin{example}
Figure \ref{fig:full-scheme} depicts a graph (a) and the meta-graph (b) and the path labelling (c) of this graph. The edge $(1,3)$ in the meta-graph is assigned with a weight $2$, i.e. $\sigma(1,3)=2$, since there is one shortest path between $1$ and $3$ which goes through $4$. The label of $4$ in the path labelling contains $(1,1)$ and $(3,1)$. The labelling entry $(2,2)$ is not included in the label of $4$ because every shortest path between $4$ and $2$ goes through another landmark, i.e. $1$ or $3$.
\end{example}

Algorithm \ref{alg:labels} describes the pseudo-code of our algorithm for constructing a labelling scheme. Given a graph $G$ and a set of landmarks $R$, we conduct a BFS from each landmark $r_i\in R$. We use two queues $Q_L$ and $Q_N$ to keep track of visited vertices, which respectively need to be labeled and not to be labeled. All vertices, except for $r_i$, are initialized as being unvisited (Line 5). For each vertex $u\in Q_L$ at the $n$-th level of the BFS, we set its unvisited neighbors $v$ being visited (Line 10). If $v$ is a landmark, we push $v$ into $Q_N$ and add an edge into $E_R$ and store the distance between $r_i$ and $v$ to the edge in $\sigma$. Otherwise, we push $v$ into $Q_L$ and add a label in $L$ for $v$ (Lines 11-17). Then, We check unvisited neighbors of each vertex $u\in Q_N$ at the $n$-th level, and push $v$ into $Q_N$ without adding a label in $L$ or an edge in $M$ (Lines 18-21). \rv{This process is conducted level-by-level on the BFS (Line 22).}

\begin{algorithm} [t] 
\caption{Constructing a labelling scheme $\mathcal{L}$}  
\label{alg:labels}
\KwIn{$G=(V,E)$; a set of landmarks $R\subseteq V$}
\KwOut{A labelling scheme $\mathcal{L}=(M,L)$ with $M=(R,E_R,\sigma)$.} 
$E_R\leftarrow \emptyset$; $L(v)\leftarrow \emptyset$ for all $v\in V$\\
\For{all $r_i\in R$}
{
    $Q_{L}\leftarrow \emptyset$;  $Q_{N}\leftarrow \emptyset$; \\
    $Q_{L}$.push($r_i$);\\
    $depth[r_i]\leftarrow 0$; $depth[v]\leftarrow \infty $ for all $v\in V\backslash\{r_i\}$;\\
    n = 0;\\
    \While{$Q_{L}$~and~$Q_{N}$~are~not~empty}
    {
        \For{all \textbf{}$u\in Q_{L}$~at~depth~n}
        {
            
            \For{all unvisited neighbors $v$ of $u$}
            {   
                $depth[v]\leftarrow n+1$;\\
                \If{$v$ is a landmark}
                {
                    $Q_{N}$.push($v$);\\
                    $E_R\leftarrow E_R\cup \{(r_i,v)\}$;\\
                    $\sigma(r_i,v)\leftarrow depth[v]$; \\
                    
                }
                \Else{
                    $Q_{L}$.push($v$);\\
                    $L(v)\leftarrow L(v) \cup \{(r_i, depth[v])\}$;\\
                }
            }       
        }
        \For{all $u\in Q_{N}$~at~depth~n}
        {
            \For{all unvisited neighbors $v$ of $u$}
            {   
                $depth[v]\leftarrow n+1$;\\
                $Q_{N}$.push($v$);\\
            }
        }
        $n\leftarrow n+1$;\\
    }
}

\end{algorithm}  

\begin{example}

\begin{figure}[t!]
    \centering
    \includegraphics[width=7.5cm]{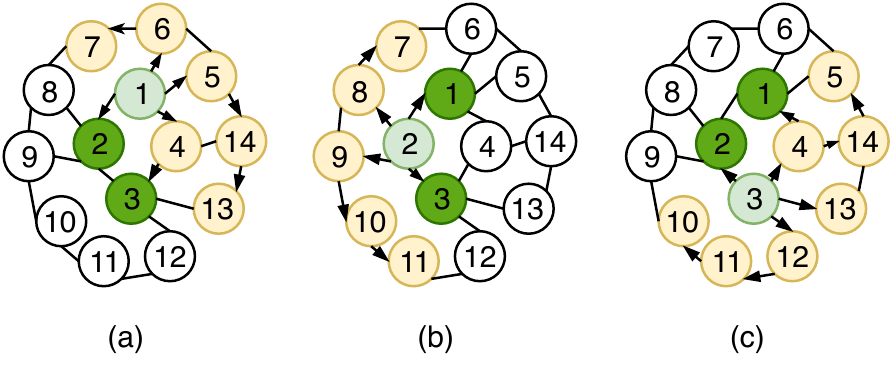}\vspace{-0.3cm}
    \caption{An illustration of labelling: (a), (b) and (c) describe the BFSs rooted at the landmarks 1, 2 and 3, respectively, where light and dark green vertices denote the landmarks, and yellow vertices denote those being labelled.}
    \label{fig:bfs}\vspace{-0.3cm}
\end{figure}
Figure \ref{fig:bfs} shows how our algorithm conducts BFSs to construct labels. The BFS from landmark $1$ is depicted in Figure \ref{fig:bfs}(a), in which vertices $\{4, 5, 6, 7, 13, 14\}$ are labelled because the other vertices are either landmarks or have landmarks in all their shortest paths to landmark $1$. We add edges $(1,2)$ and $(1,3)$ into the meta-graph. In the BFS from landmark $2$ in Figure \ref{fig:bfs}(b), vertices $\{7,8,9,10,11\}$ are labelled because the shortest paths between $2$ and vertices in $\{4,5,6,12,13,14\}$ all go through landmark $1$ or $3$. 
The BFS from landmark $3$ is depicted in Figure \ref{fig:bfs}(c), which works in a similar manner.

\end{example}

%% file: subsection_sketch.tex
\begin{figure*}
    \centering
    \includegraphics[width=17cm]{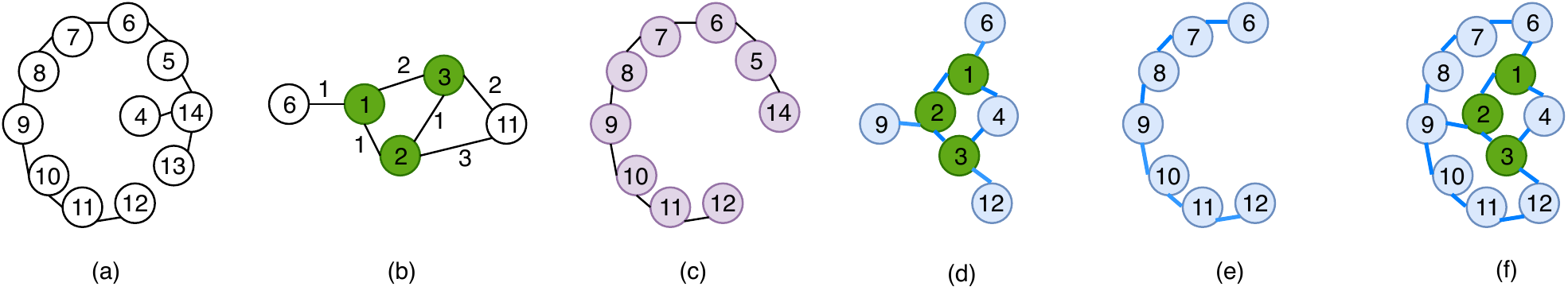}\vspace{-0.2cm}
    \caption{An illustration of sketching and searching: (a) the sparsified graph $G^-$ of the graph $G$ shown in Figure \ref{fig:full-scheme}(a); (b) the sketch for \rv{SPG(6,11)} on the graph $G$; (c) the bi-directional BFS on $G^-$, (d) the recover search based on $\mathcal{L}$, (e) the reverse search based on $G^-$, and (f) shows the query answer of \rv{SPG(6,11)}.} 
    \label{fig:example-search}\vspace{-0.1cm}
\end{figure*}

\subsection{Fast Sketching}
\label{subsec:sketch}
Let $\mathcal{L}=(M, L)$ be a labelling scheme on a graph $G$. For a given query $SPG(u,v)$, we proceed to answer $SPG(u,v)$ in two steps; (1) computing a sketch for two vertices $u$ and $v$ from the labelling scheme $\mathcal{L}$ efficiently; (2) computing the exact answer by conducting a guided search based on the sketch for two vertices $u$ and $v$. Hence, the purpose of such a sketch is to provide an efficient and principled way of searching the answer of $SPG(u,v)$, which is particularly important on very large networks. 

\begin{definition}{\emph{[\textsc{Sketch}]}}\label{def:sketch}
A \emph{sketch} for \rv{$SPG(u,v)$} on $\mathcal{L}$ is $S_{uv}=(V_S,E_S,\sigma_S)$ where $V_S= \{u,v\}\cup R$ is a set of vertices, $E_S$ is a set of edges, and $\sigma_S : E_S\mapsto \mathbb{N}$ with $\sigma_S (u',v')=d_G(u',v')$, satisfying the condition that $E_S$ contains only edges lying on the paths between $u$ and $v$ with the minimal length as defined below:
\begin{align}\label{equ:sketch}
\begin{split}
\rv{d_{uv}^{\top}}=\min_{(r,r')}\{\delta_{ru}+d_M(r,r')+\delta_{r'v}|(r,\delta_{ru})\in L(u),\\[-8pt] (r',\delta_{r'v})\in L(v)\}; 
\end{split}
\end{align}
\end{definition}


\rv{Accordingly, we have the following corollary.}

\begin{corollary}
$\rv{d_{uv}^{\top}}\geq d_G(u,v)$ holds.

\end{corollary}



 Algorithm \ref{alg:sketch} describes how to construct a sketch. Let $u$ and $v$ be a pair of vertices. We start with $V_S=\emptyset$ and $E_S=\emptyset$. Then, for each pair of landmarks $\{r, r'\}$, we compute the minimum length $\pi_{rr'}$ of paths between $u$ and $v$ that go through $r$ and $r'$ using the labels in $L$ and the meta graph $M$ (Lines 2-5). After that, we obtain the minimum length of paths between $u$ and $v$ that go through at least one landmark, \rv{i.e., $d_{uv}^{\top}$} (Line 6), and add the edges in these paths into $E_S$, the vertices in these paths into $V_S$, and the corresponding distances are associated with the edges \rv{(Lines 7-13)}.


\begin{algorithm}  
\caption{Computing a sketch $S_{uv}$}  
\label{alg:sketch}
\KwIn{$\mathcal{L}=(M,L)$, two vertices $u$ and $v$.}
\KwOut{A sketch  $S_{uv}=(V_S,E_S,\sigma_S)$} 
$V_S\leftarrow \emptyset$, $E_S\leftarrow \emptyset$;\\
\For{all $\{r, r'\}\subseteq R$}
{
    $\pi_{rr'}\leftarrow +\infty$; \\
    \If{$(r,\delta_{ur})\in L(u)$ and $(r',\delta_{vr'})\in L(v)$}
    {
        $\pi_{rr'}\leftarrow \delta_{ur} + d_M(r,r') + \delta_{vr'}$;\\
    }
}
$d_{uv}^{\top} \leftarrow$ min\{$\pi_{rr'} | \{r,r'\}\subseteq R$\};\\
\For{all $\{r,r'\}\subseteq R$ and $\pi_{rr'}=d_{uv}^{\top}$}
{
 $E_S\leftarrow E_S \cup \{(u,r),(v,r')\}$;\\
        $\sigma_S(u,r)\leftarrow \delta_{ur}$, $\sigma_S(v,r')\leftarrow \delta_{vr'}$;\\
        \For{all $(r_i,r_j)$ in the shortest path graph of $(r,r')$ in $M$}
        {   
            $E_S\leftarrow E_S \cup \{(r_i,r_j)\}$;\\ 
            $\sigma_S(r_i,r_j)\leftarrow \sigma(r_i,r_j);$\\
        }
                    $V_S\leftarrow V(E_S)$;\\
}
\end{algorithm}  

\begin{example}
Figure \ref{fig:example-search}(b) shows the sketch between two vertices $6$ and $11$. The sketch has the edges $(1,6)$, $(1,3)$, $(3,11)$, $(2,3)$ $(1,2)$ and $(2,11)$ because we have the following shortest paths between $6$ and $11$ with $\delta_{6,1}+d_M(1,3)+\delta_{11,3}=5$ 
and $\delta_{6,1}+d_M(1,2)+\delta_{11,2}=5$. We thus have $\rv{d^{\top}_{6,11}}=5$, and $\rv{d^{\top}_{6,11}}=d_G(6,11)$. 
\end{example}

%% file: subsection_search.tex
\subsection{Guided Searching}
\label{subsec:searching}
Guided by $S_{uv}$, we conduct a search to compute the exact answer of \rv{$SPG(u,v)$}, based on the following observations: 

\begin{itemize}
	\item Such a search can be conducted on a sparsified graph $G[V\backslash R]$ by removing all landmarks in $R$ and all edges incident to these landmarks from $G$. $d_{G[V\backslash R]}(u,v)$ may potentially be greater than $d_G(u,v)$; however, the number of search steps in this sparsified graph can be upper bounded by \rv{$d_{uv}^{\top}$ due to the fact that $d_G(u,v)=min(d_{G[V\backslash R]}(u,v), d_{uv}^{\top})$.} 
	
	\item $S_{uv}$ can guide how to conduct a bi-directional search on the sparsified graph $G[V\backslash R]$. Specifically, for $t\in\{u,v\}$, we have 
	\begin{align}\label{equ:sidechoice}
		d_t^{*}=\max_{(r,t)\in E_S}\sigma_S(r,t)-1,
	\end{align}
	\rv{which suggests the number of search steps from the $u$ and $v$ sides, respectively.} 
	Here, we subtract 1 because $r$ can be found via labels of vertices in at most $\sigma_S(r,t)-1$ steps.
\end{itemize}

Given a query $SPG(u,v)$ on a graph $G$, the answer $G_{uv}$ can thus be computed by searching over the sparsified graph $G^-=G[V \backslash R]$ and the labelling scheme $\mathcal{L}$, guided by the sketch $S_{uv}$, as follows:   

\begin{equation}\label{equ:search}
	G_{uv}= 
	\begin{cases}
		G^{\mathcal{L}}_{uv}& \text{if } d_{G^-}(u,v)> d^{\top}_{uv};\\
		G^-_{uv}\cup G^{\mathcal{L}}_{uv}& \text{if } d_{G^-}(u,v)=d^{\top}_{uv};\\
		G^-_{uv}           & \text{otherwise}.
	\end{cases}
\end{equation}
We use $G^{\mathcal{L}}_{uv}$ to refer to shortest paths between $u$ and $v$ that go through at least one landmark in $R$. 

\rv{Generally, a guided search has three stages: (1) \emph{Bi-directional search}, which has a \emph{forward search} from the $u$ side and a \emph{backward search} from the $v$ side \cite{goldberg2005computing}, under the guide of $S_{uv}$ w.r.t. Eq. \ref{equ:sidechoice}. This search terminates when common vertices are found or the upper bound $d^{\top}_{uv}$ is reached. 
(2) \emph{Reverse search}, which reverses the previous bi-directional search back to $u$ and $v$ in order to compute shortest paths in $G^-_{uv}$. 
(3) \emph{Recover search}, which recovers the relevant labelling information under the guide of $S_{uv}$ in order to compute shortest paths in $G^{\mathcal{L}}_{uv}$.}
\rv{As we do not know initially which of the three cases of Eq. \ref{equ:search} holds, a bi-directional search is always performed. This search provides us with $d_{G^-}(u,v)$, though we abort once $d_{G^-}(u,v)>d_{uv}^{\top}$ can be guaranteed.}
Then depending on the values of $d_{G^-}(u,v)$ and $d^{\top}_{uv}$, a reverse search, a recover search, or both of them are performed to compute $G^-_{uv}$ and $G^{\mathcal{L}}_{uv}$ as in Eq. \ref{equ:search}.

\begin{algorithm}[t]
	\caption{Searching on $G[V\backslash R]$}
	\label{alg:bounded-search}
	\KwIn{$G^-=G[V\backslash R]$, $S_{uv}$, $\mathcal{L}=(M,L)$}
	\KwOut{A shortest path graph $G_{uv}$}
	$d_{uv}^{\top}, d_u^{*}, d_v^{*} \leftarrow get\_bound(S_{uv})$;\\
	$P_u \leftarrow \emptyset$, $P_v \leftarrow \emptyset$, $d_u\leftarrow 0$, $d_v \leftarrow 0$;\\
	Enqueue $u$ to $Q_u$ and $v$ to $Q_v$;\\
	$depth_u[w]\leftarrow \infty$, $depth_v[w]\leftarrow \infty$ for all $w\in V\backslash R$;\\
	$depth_u[u]\leftarrow 0$, $depth_v[v]\leftarrow 0$;\\
	\While{$d_u+d_v<d_{uv}^{\top}$}
	{
		$t\leftarrow pick\_search(P_u,P_v,d_u^{*},d_v^{*},d_u,d_v)$;\\
		
		\If{$t=u$} 
		{
			$Q_u\leftarrow forward\_search(Q_u$);\\ 
		}
		\If{$t=v$}
		{
			$Q_v\leftarrow backward\_search(Q_v$);\\
		}
		\rv{$P_t\leftarrow P_t\cup Q_t$}; $d_t\leftarrow d_t+1$;\\
		$depth_t[w]\leftarrow d_t$ \textbf{for} $w\in Q_t$;\\    
		\If{$P_u\cap P_v$ is not empty}
		{
			\textbf{break};\\
		}
	}
	\If{\rv{$P_u \cap P_v\neq \emptyset$}}
	{
		\rv{$G^-_{uv}\leftarrow reverse\_search(P_u\cap P_v,G^-,depth_u,depth_v)$;}\\
	}
	
	\If{$d_u+d_v=d_{uv}^{\top}$}
	{
		$Z \leftarrow \emptyset$;\\
		\For{all $\rv{(r,t)}\in E_S$ with $t\in \{u,v\}$}
		{
		    $d_m \leftarrow \min\{\sigma_S(r,t)-1,d_t\}$;\\
		    \For{all $w$ with $depth_t[w]=d_m$, $(r,\delta_{wr})\in L(w)$, $\delta_{wr}+d_m=\sigma_S(r,t)$} 
		    {
		    $Z\leftarrow Z \cup \{(w,r)\}$;\\
		    }
		}
		\rv{$G^{\mathcal{L}}_{uv}\leftarrow recover\_search(S_{uv}, \mathcal{L}, Z, G^-, depth_u, depth_v$});\\
	}
	\rv{$G_{uv}\leftarrow G^-_{uv} \cup G^{\mathcal{L}}_{uv}$;}\\
	
\end{algorithm}

Algorithm \ref{alg:bounded-search} presents our guided search algorithm. We maintain two queues $P_u$ and $P_v$ which contain the set of all vertices traversed from $u$ and $v$, respectively. $d_u$ and $d_v$ indicate the levels of traversal being conducted in the BFSs rooted at $u$ and $v$, respectively. \rv{Two queues $Q_u$ and $Q_v$ keep vertices being searched} from $u$ and $v$ at the $d_u$ and $d_v$ level, respectively.
Initially $P_u$ and $P_v$ are empty, and $u$ and $v$ are enqueued into $Q_u$ and $Q_v$ respectively.
$depth_u$ and $depth_v$ denote the depths of all vertices in the BFSs rooted at $u$ and $v$. 

\rv{A bi-directional search is first conducted (Lines 6-15).} In each iteration, the bi-directional search is guided by $d_u^*$ and $d_v^*$ as well as the relative sizes of $P_u$ and $P_v$ to decide the next step \rv{(Line 7).} 
\rv{We choose $t$ where $d_t^*>d_t$ and $t\in\{u,v\}$. If both $u$ and $v$ satisfy this condition, or none of them satisfy this condition, then the choice of a forward search ($t=u$) and a backward search ($t=v$) is determined by the sizes of $P_u$ and $P_v$.} Accordingly, $P_u$ or $P_v$ are expanded \rv{(Line 12)}. The bi-directional search terminates either when $d_u+d_v$ reaches the upper bound $d_{uv}^{\top}$ or $P_u\cap P_v$ is not empty.
\rv{This approach extends the \emph{Optimized Bidirectional BFS} algorithm of \cite{hayashi2016fully} by incorporating bounds obtained from our sketch.}

\rv{If $P_u\cap P_v$ is not empty, we have $d_{G^-}(u,v)\leq d_{uv}^{\top}$ and thus start a reverse search (Lines 16-17). For each vertex $x\in P_u\cap P_v$, we compute the shortest paths between $u$ and $x$ and between $v$ and $x$ according to the depths of vertices in $depth_u$ and $depth_v$, respectively. For example, a neighbour $x'$ of $x$ in $G^-$ is on the shortest path between $x$ and $u$ if $depth_u[x]-1=depth_u[x']$, and thus we find such $x'$ and compute shortest paths between $x'$ and $u$ in the same manner. 
If $d_u+d_v=d_{uv}^{\top}$, we have $d_{G^-}(u,v)\geq d_{uv}^{\top}$ and start a recover search (Lines 18-24). For each edge $(r,t)$ in the sketch $S_{uv}$ and $t\in \{u,v\}$, we search for all vertices $w$ with $depth_t[w]=\min\{\sigma_S(r,t)-1,d_t\}$ and $\sigma_S(r,t)=\delta_{wr} + depth_t[w]$ (Lines 19-23).} \rv{Each $w$ is a vertex closest to landmark $r$ among all vertices on at least one shortest path between $r$ and $t$ in our previous bi-directional search. $Z$ stores $(w,r)$ pairs to guide the recover searches.
In the recover search (Line 24), for each edge $(r,r')$ in $S_{uv}$ where $r, r'\in R$, we recover the shortest paths between $r$ and $r'$ according to $\mathcal{L}$. For each $(w,r)\in Z$, we find shortest paths between $w$ and $r$ according to $G^-$ and labelling information $\mathcal{L}$. For example, for a neighbour $w'$ of $w$ in $G^-$, $w'$ is on the shortest path between $w$ and $r$ if $(r, \delta_{w'r})\in L(w')$ and $\delta_{w'r}+1=\delta_{wr}$. The shortest paths between $w$ and $u$ (resp. $v$) is computed according to $depth_u[~]$ (resp. $depth_v[~]$), but the search for parts of shortest paths that have already been found in the reversed search can be skipped.}
\rv{We also compute the shortest paths between relevant landmarks.}

\vspace{-0.1cm}
\begin{example}
	Figure \ref{fig:example-search}(c)-(e) illustrates how our guided searching finds the answer for a query \rv{SPG(6,11)}. The sparsified graph $G^-$ is depicted in Figure \ref{fig:example-search}(a) and the sketch is depicted in Figure \ref{fig:example-search}(b). The sketch provides the upper bound $d^\top_{6,11}=5$, $d^*_6=0$ and $d^*_{11}=2$ because $\rv{\sigma_S(1,6)}=1$ and $\rv{\sigma_S(2,11)}=3$, respectively. The bi-directional BFS is depicted in Figure \ref{fig:example-search}(c), in which $d_6=2$, $d_{11}=3$, $P_6=\{5,7,8,14\}$, and $P_{11}=\{10,12,9,8\}$. The queues $P_6$ and $P_{11}$ meet at vertex $8$, and thus $d_{G^-}(6,11)=5$. \rv{The reverse search is depicted in Figure \ref{fig:example-search}(e), which goes back to $6$ and $11$ from $P_6\cap P_{11}=\{8\}$. The recover search is depicted in Figure \ref{fig:example-search}(d), which finds shortest paths going through the landmarks $\{1,2,3\}$ with $Z=\{(12,3),(9,2),(6,1)\}$ and recovers shortest paths between landmarks in the sketch.} The final query answer is depicted in Figure \ref{fig:example-search}(f).
	
\end{example}

%% file: section_theoretical-analysis.tex
\section{Theoretical Discussion}
\label{sec:proof}
We prove the correctness of QbS and analyze its complexity. We also discuss how to \rv{parallelize} the labelling construction process.

\begin{table*}[ht]
    \centering
    \begin{tabular}{l||ll|rrr|rrr|r}
    \toprule
        Dataset & Network & Type & $|V|$ & $|E|$ & $|E^{un}|$  & max. deg & avg. deg & avg. dist & $|G|$\\
        \midrule
        Douban (DO) & social & undirected & 0.2M & 0.3M & 0.3M & 287 & 4.2 & 5.2 &  2.5MB\\ 
        DBLP (DB) & co-authorship & undirected & 0.3M & 1.1M & 1.1M & 343 & 6.6 & 6.8 & 8.0MB\\
        \midrule
        Youtube (YT) & social & undirected & 1.1M & 3.0M & 3.0M & 28,754 & 5.27 & 5.3 & 23MB\\
        WikiTalk(WK) & communication & directed & 2.4M & 5.0M & 4.7M & 100,029 & 3.89 & 3.9 & 36MB\\
        Skitter (SK) & computer & undirected & 1.7M & 11.1M & 11.1M & 35,455 & 13.08 & 5.1 & 85MB\\
        Baidu (BA) & web & directed & 2.1M & 17.8M & 17.0M & 97,848 & 15.89 & 4.1 & 130MB\\
        LiveJournal (LJ) & social & directed & 4.8M & 68.5M & 43.1M & 20,334 & 17.79 & 5.5 & 329MB\\
        Orkut (OR) & social & undirected & 3.1M & 117M & 117M & 33,313 & 76.28 & 4.2 & 894MB\\
        \midrule
        Twitter (TW) & social & directed & 41.7M & 1.5B & 1.2B & 2,997,487 & 57.74 & 3.6 & 9.0GB\\
        Friendster (FR) & social & undirected & 65.6M & 1.8B & 1.8B & 5,214 & 55.06 & 4.8 & 13.0GB\\
        uk2007 (UK) & web & directed & 106M & 3.7B & 3.3B & 979,738 & 62.77 & 5.6 & 24.8GB\\
        ClueWeb09 (CW) & computer & directed & 1.7B & 7.8B & 7.8B & 6,444,720 & 9.27 & 7.5 & 58.2GB\\
        \bottomrule
    \end{tabular}
    \caption{Datasets, where $|E^{un}|$ is the number of edges in a graph being treated as undirected, and $|G|$ denotes the size of a graph $G$ with each edge appearing in the adjacency lists and being represented by 8 bytes.}
    \label{tab:datasets}\vspace*{-0.4cm}
\end{table*}

\subsection{Proof of Correctness}

In the following, we prove the theorem for the correctness of QbS.
\begin{theorem}
Given any query \rv{SPG(u,v)} on a graph $G$, the answer $G_{uv}$ can be computed using QbS.
\end{theorem}
\begin{sproof}
We first prove that a labelling scheme constructed by Algorithm \ref{alg:labels} satisfies Definition \ref{def:scheme}. Suppose that we conduct a BFS rooted from $r\in R$. Given a landmark $r'\in R\backslash \{r\}$, if $\exists p\in P_{rr'} (V(p)\cap R=\{r,r'\})$ holds, there must exist $w\in Q_L$ with $depth[w]+1=depth[r']$ and $(w,r')\in E$ (Lines 8-9, 11), and accordingly an edge $(r,r')$ is added into $M$ (Lines 13-14). Otherwise, $r'$ is directly pushed into $Q_N$ (Lines 19-21). Given a vertex $v\in V\backslash R$ that is not a landmark, if $\exists p\in P_{rv} (V(p)\cap R=\{r\})$ holds, there must exist $w\in Q_L$ with $depth[w]+1=depth[v]$ and $(w,v)\in E$ (Lines 8-9, 15), and accordingly a label $(r,depth[v])$ is added into $L$ (Lines 16-17). Otherwise, $v$ is directly pushed into $Q_N$ (Lines 19-21). 

Now we prove that a sketch constructed by Algorithm \ref{alg:sketch} satisfies Definition \ref{def:sketch}. First, Algorithm \ref{alg:sketch} (Lines 2-7) finds pairs of landmarks $(r,r')$ that minimise $\{\delta_{ur} + d_M(r,r') +\delta_{r'v}|\rv{(r,\delta_{ur})}\in L(u)$ and $(r',\delta_{r'v})\in L(v)\}$ (i.e., satisfying Eq. (3) in Definition \ref{def:sketch}). Then it adds $(u,r),(r',v)$ and all edges on the shortest paths between $(r,r')$ on a meta-graph into the sketch (Lines 8-12). 

Finally, we prove that $G_{uv}$ can be constructed by Algorithm \ref{alg:bounded-search}. Each shortest path between $u$ and $v$ that does not go through any landmark can be constructed from $G^-$ using a bi-directional BFS and its reverse search (Lines 6-15 \rv{and 16-17}). For each shortest path between $u$ and $v$ that goes through at least one landmark, \rv{all such landmarks must be included in $S_{uv}$ and such shortest paths are computed using the recover search (Lines 18-24)}.
\end{sproof}

\subsection{Complexity Analysis}\label{subsec:complexity} The time complexity of constructing a BFS from one landmark in Algorithm \ref{alg:labels} is $O(|E|)$ and the overall time complexity of Algorithm \ref{alg:labels} is $O(|R||E|)$.  The time complexity of constructing a sketch in Algorithm \ref{alg:sketch} is $O(|R|^4)$ and can be reduced to $O(|R|^2)$ by precomputing shortest path distances and shortest paths between landmarks on a meta-graph constructed by Algorithm \ref{alg:sketch}, i.e., computation on Lines 10-12 is saved. The time complexity of conducting a guided search in Algorithm \ref{alg:bounded-search} is $O(|E|+|R||V|)$. 

Note that, in our work, the number of landmarks is small, i.e., $|R|=20$ by default, which is much smaller than the number of vertices or edges in the original graph. Thus, we can see that, constructing a labelling scheme by Algorithm \ref{alg:labels} is indeed $O(|E|)$, computing a sketch is constant time, and performing a guided search becomes $O(|E^*|+|V|)$ where $|E^*|$ denotes the number of edges in the sparsified graph after removing edges incident to landmarks from $G$.

\subsection{Parallelization} 

Given a graph $G$ and a set of landmarks $R$ in $G$, a nice property of our labelling scheme $\mathcal{L}$ is that there is only one such labelling scheme. Formally, we prove the lemma below.
\begin{lemma}\label{lem:deterministic}
Let $\mathcal{L}$ be a labelling scheme on a graph $G$ w.r.t. a set of landmarks $R$. $\mathcal{L}$ is deterministic.
\end{lemma}
\begin{sproof}
A labelling scheme $\mathcal{L}$ \rv{consists of a meta-graph $M=(R,E_R,\sigma)$ and a path labelling $L$.} 
From Definition \ref{def:metaG}, an edge $(r,r')\in E_R$ if and only if there exists at least one shortest path between $r$ and $r'$ that does not go through any other landmarks in $R\backslash \{r,r'\}$. From Definition \ref{def:scheme}, a label $(r,\delta_{ur})\in L(u)$ if and only if there exists at least one shortest path between $u$ and $r$  that does not go through any other landmarks in $R\backslash\{r\}$. Therefore, $\mathcal{L}$ is deterministic w.r.t $G$ and $R$.
\end{sproof}

For a fixed set of landmarks, the labelling construction in Algorithm \ref{alg:labels} yields the same labelling scheme, regardless of the ordering of landmarks. This deterministic nature of labelling scheme enables us to speed up the construction of labelling scheme by paralleling Algorithm \ref{alg:labels}. If we use one thread for constructing labels from one landmark, then we can leverage the thread-level parallelism to perform BFSs from different landmarks simultaneously.

%% file: section_experiment.tex
\section{Experiments}\label{sec:experiments}
We evaluated our method $QbS$ to answer the following questions: 
\begin{itemize}
\item[(\textbf{Q1})] How efficiently can our proposed method answer \rv{shortest-path-graph} queries, while still achieving construction time efficiency and low labelling space overhead? 
\item[(\textbf{Q2})] How well can sketching help improve the performance of answering \rv{shortest-path-graph} queries? 
\item[(\textbf{Q3})] How does the number of landmarks affect the performance such as construction time, labelling size and query time?
\end{itemize}

\subsection{Experimental Setup}
We implemented our proposed methods in C++ 11 and compiled using g++. We performed all experiments on a Linux \rv{server} which has Intel Xeon W-2175 with 2.5GHz and 512GB of main memory.
\vspace*{-0.4cm}
\begin{figure}[h]
    \includegraphics[width=8cm]{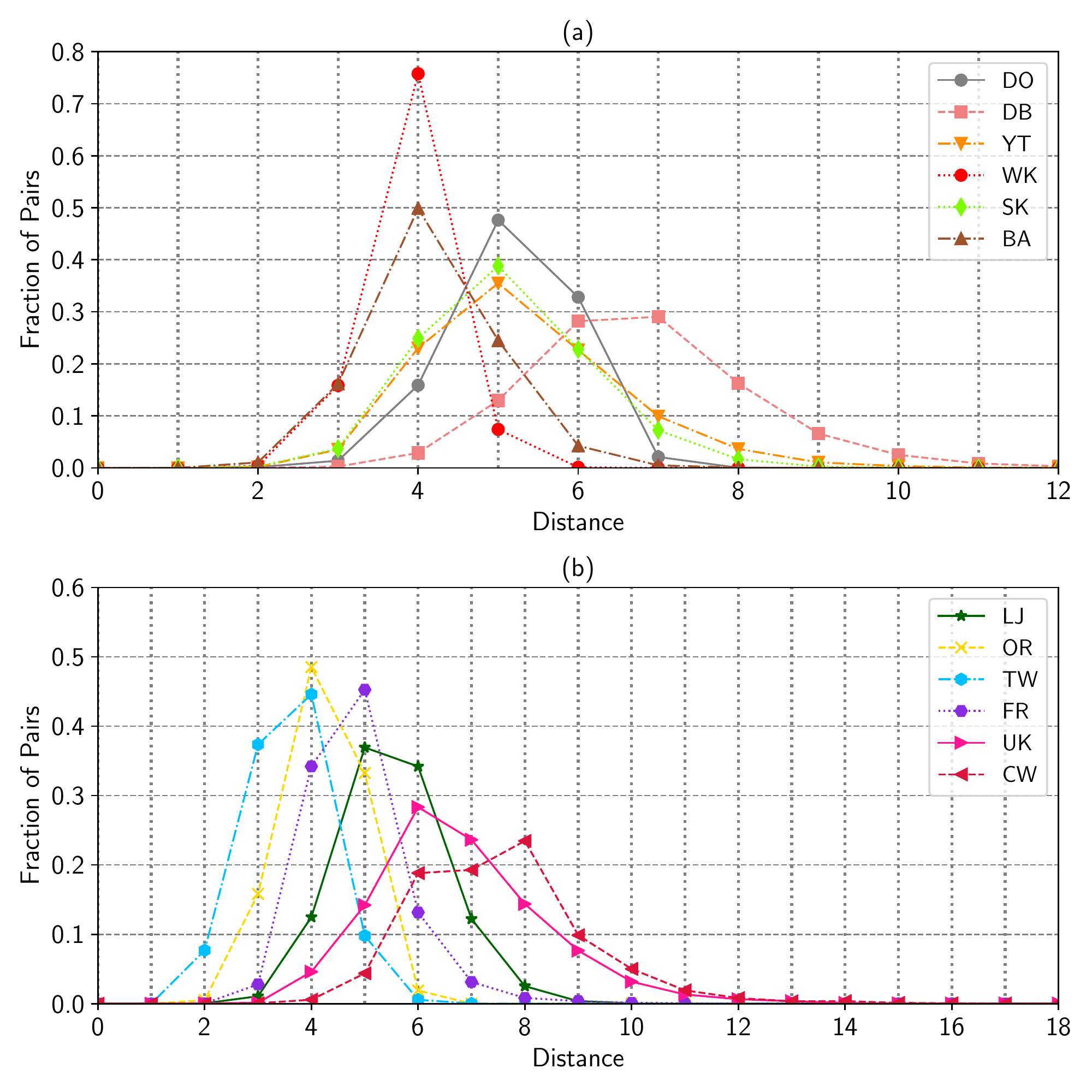}\vspace{-0.3cm}
    \caption{Distance distribution of 10,000 randomly selected pairs of vertices on all the datasets.}
    \label{fig:distance-distribution}\vspace{-0.5cm}
\end{figure}

\medskip
\noindent\textbf{Datasets. }We conducted experiments on 12 real-world graph datasets from various types of complex large networks, including social networks, computer networks, web networks, co-authorship networks and communication networks. Table \ref{tab:datasets} presents the details of these datasets, among which the largest one has 1.7 \rv{billion} vertices and 7.8 \rv{billion} edges. We treated graphs in these datasets as being undirected. All the datasets used in our experiments are publicly available from Koblenz Network Collection \cite{kunegis2013konect}, Stanford Networks Analysis Project \cite{leskovec2014snap}, Dynamically Evolving Large-scale Information Systems Project \footnote{See http://law.di.unimi.it/datasets.php for datasets} and the Lemur Project\footnote{See https://lemurproject.org/clueweb09/index.php}.

\begin{table*}[ht!]
    \centering
    \begin{tabular}{l||rr|rr||r|rr|r}
        \toprule
        \multirow{2}*{Dataset \hspace{1cm}}& \multicolumn{4}{c||}{Construction Time (sec.)} & \multicolumn{4}{c}{\rv{Average Query Time} (ms.)}\\
        \cline{2-9}
         & QbS-P & QbS & \hspace{0.3cm} PPL & ParentPPL  & \hspace{0.5cm}  QbS & PPL & ParentPPL & \hspace{0.5cm} Bi-BFS \\
        \midrule
         Douban & 0.05 & 0.3 & 154 & 2,736 & 0.037 & 1.414 & 0.038 & 0.585 \\ 
         DBLP & 0.12 & 1.1 & 2,610 & 11,049 & 0.097 & 1.782 & 0.052 & 2.995  \\ \hline
         Youtube & 0.47 & 4.4 & 22,601 & DNF & 0.218 & 5.314 & - & 23.809\\ 
         WikiTalk & 0.61 & 4.9 & 8,662 & DNF & 0.693 & 3.536 & - & 6.984 \\ 
         Skitter & 1.51 & 12.7 & 86,326 & DNF & 0.951 & 16.978 & - & 44.685 \\ 
         Baidu & 2.04 & 18.9 & DNF & \rv{OOE} & 0.845 & - & - & 174.412\\ 
         LiveJournal & 6.48 & 52.2 & DNF & \rv{OOE}  & 1.095  & - & - & 84.967 \\ 
         Orkut & 10.85 & 73.2 & DNF & \rv{OOE} & 4.237 & - & - & 207.541 \\\hline
         Twitter & 199.8 & 1,345 & DNF & \rv{OOE}  & 164.333 & - & - & 4,817.774 \\ 
         Friendster & 416.5 & 2,354 & DNF & \rv{OOE} & 11.972 &  - & - & 3,600.362 \\  
         uk2007 & 178.5 & 1,485 & \rv{OOE} & \rv{OOE} & 77.830 & - & - & 5,264.101 \\ 
         ClueWeb09 & 1,819 & 17,060 & \rv{OOE} & \rv{OOE} & 480.443 & - & - & \rv{DNF} \\ 
        \bottomrule
    \end{tabular}
    \caption{\rv{Comparison of construction time and query time. DNF and OOE refer to running out of time (>24 hours) and running out of memory, respectively.} }
    \label{tab:compare_construction_time}\vspace*{-0.7cm}
\end{table*}

\medskip
\noindent\textbf{Queries. }We randomly sampled 10,000 pairs of vertices from all pairs of vertices in each graph to evaluate the average query time. 
Figure \ref{fig:distance-distribution} shows the distance distribution  of these 10,000 randomly sampled pairs of vertices in each graph dataset. We can see that the distances of these pairs of vertices mostly fall into the range of 2-9.

\medskip
\noindent\textbf{Baselines. } 
We considered the following baselines:
\begin{itemize}
\item[(1)] \textbf{Labelling-based methods}. Pruned landmark labelling (PLL) is the state-of-the-art method for computing exact distance queries \cite{akiba2013fast}. We thus use the methods \emph{Pruned Path Labelling} (PPL) and \emph{Pruned Path Labelling with Parent information} (ParentPPL) as discussed in Section \ref{sec:labelling-based-methods} as our baselines. 
\item[(2)] \textbf{Search-based methods.} We use bi-directional BFS as the baseline which conducts search from the directions of two vertices alternatively \cite{goldberg2005computing}. We denote it as Bi-BFS.
\end{itemize}\vspace*{-0.1cm}
To evaluate the parallel speed-up of construction time, we use QbS to refer to our method with a sequential labelling construction and QbS-P to refer to our method with a parallel labelling construction, with up to 12 threads in our experiments. In PPL and ParentPPL, we use 32 bits and 8 bits to represent a landmark and a distance in their labels, respectively, and 32 bits to store each parent in ParentPPL. In QbS and QbS-P, we use $|R|$*8 bits to store the label of each vertex.

\medskip
\noindent \textbf{Landmarks.}
In PPL and ParentPPL, landmarks are ordered in descending order of degrees. In QbS, we choose vertices with the largest degrees as landmarks for two reasons: (1) removing high-degree vertices sparsifies a graph much more than low-degree vertices; (2) computing distances from two vertices to high-degree landmarks provides a good estimation of the shortest distance between these two vertices \cite{potamias2009fast}. We set $|R|=20$ in QbS by default.



\subsection{Performance Comparison}
We conducted experiments to compare construction time, labelling size and query time of our method against the baselines. 

\subsubsection{Construction Time}

Table \ref{tab:compare_construction_time} shows that our method QbS can efficiently construct a labelling scheme on all the datasets, scaling over 
large networks with billions of vertices and edges. Compared with PPL and ParentPPL, our method QbS uses a significantly less amount of time (i.e., 2-4 orders of magnitude faster) to construct labelling information. Moreover, PPL failed to construct labels for 7 out of 12 datasets and ParentPPL failed for 10 out of 12 datasets. This is because these methods need to meet the 2-hop path cover property. The reason why ParentPPL is much slower than PPL is because a vertex often has more than one parent and finding all parents takes more time though the time complexity remains unchanged. We can also see that, compared with QbS, QbS-P can further improve construction time (i.e., 6-12 times faster), leading to much better scalability than QbS.

\subsubsection{Labelling Size }\label{subsubsec:labellingsize-20}

Table \ref{tab:label} presents the comparison results for the labelling sizes of QbS, PPL and ParentPPL on all the datasets. \rv{We use $size(\Delta)$ to denote the size of precomputed shortest path graphs between landmarks as discussed in Section \ref{subsec:complexity}.} We observe that: 1) the labelling sizes of QbS are hundreds of times smaller than the labelling sizes of PPL and ParentPPL; 2) the labelling sizes of ParentPPL are about twice as the labelling sizes of PPL. For dense graphs, such as Twitter, the sizes of precomputed shortest paths in QbS are relatively larger than the ones in sparse graphs. This is due to the existence of many shortest paths between landmarks in dense graphs. Nonetheless, it is important to notice that, the sizes of precomputed shortest paths between landmarks (i.e. \rv{$size(\Delta)$} in Table \ref{tab:label}) are small in QbS, compared with the sizes of labelling (i.e. \rv{$size(\mathcal{L})$} in Table \ref{tab:label}). For meta-graphs, since 
each meta-graph contains at most $|R|^2$ edges, the space overhead for storing edges and weights of a meta-graph is very small. Indeed, even when we have $|R|$=100, the size of a meta-graph would still be smaller than 0.01MB. 
In summary, these results show that QbS can scale well over very large networks in terms of the labelling size.

\begin{table}[ht!]
    \begin{tabular}{l||rr|rr}
        \toprule
         \multirow{2}*{Dataset} & \multicolumn{2}{c|}{QbS} &\multirow{2}*{PPL}&\multirow{2}*{ParentPPL} \\
        \cline{2-3}
          & \rv{$size(\mathcal{L})$}\hspace{0cm} & \rv{$size(\Delta)$}\hspace{0cm} &  &\\
        \midrule
          Douban& 2.95MB & 0.03MB  & 0.4GB & 0.8GB\\
          DBLP & 6.05MB & 0.03MB & 1.2GB & 2.4GB \\\hline
          Youtube  & 21.6MB & 0.6MB & 1.7GB & $-$ \\
          WikiTalk & 45.7MB & 0.7MB & 2.1GB & $-$ \\
          Skitter & 32.4MB & 20.3MB & 9.2GB & $-$ \\
          Baidu & 40.8MB & 4.8MB & $-$ & $-$ \\
          LiveJournal & 92.5MB & 1.1MB & $-$ & $-$ \\
          Orkut & 58.6MB & 3.5MB & $-$ & $-$ \\\hline
          Twitter & 0.78GB & 0.76GB & $-$ & $-$ \\
          Friendster & 1.22GB & 0.01GB & $-$ & $-$ \\
          uk2007 & 1.98GB & 0.08GB & $-$ & $-$ \\
          ClueWeb09 & 31.4GB & 0.48GB & $-$ & $-$ \\
        \bottomrule
    \end{tabular}
    \caption{Comparison of labelling sizes. \rv{$size(\mathcal{L})$ denotes the size of a labelling scheme $\mathcal{L}$ and $size(\Delta)$ the size of precomputed shortest-path graphs between landmarks in QbS. 
    } 
   }\label{tab:label}\vspace*{-0.5cm}
\end{table}

\begin{figure*}[ht]
    \centering
    \includegraphics[width=17.5cm]{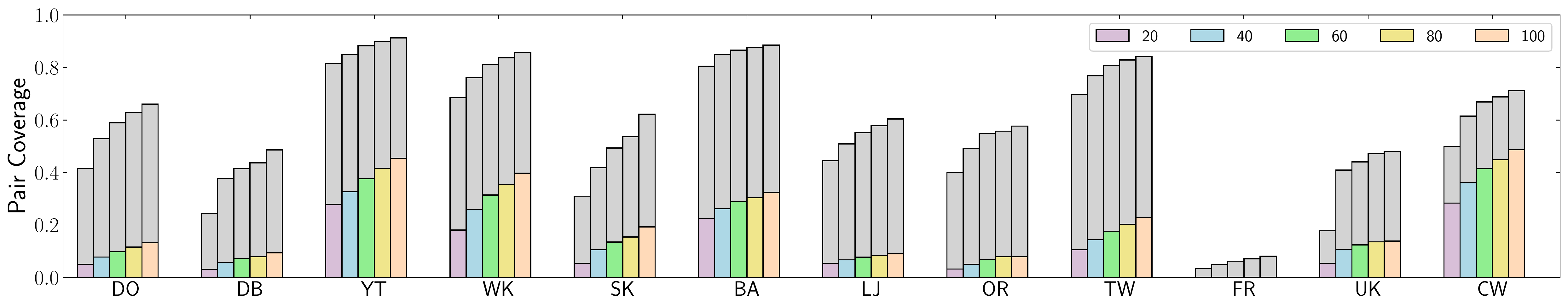}\vspace{-0.3cm}
    \caption{Pair coverage ratios using our method QbS under 20-100 landmarks where light color denotes the ratio of all the shortest paths between a vertex pair go through landmarks and grey color denoted the ratio of some but not all shortest paths between a vertex pair go through landmarks. }
    \label{fig:coverage}\vspace{-0.3cm}
\end{figure*}
\subsubsection{Query Time }

\begin{figure*}[ht]
    \centering
    \includegraphics[width=17.5cm]{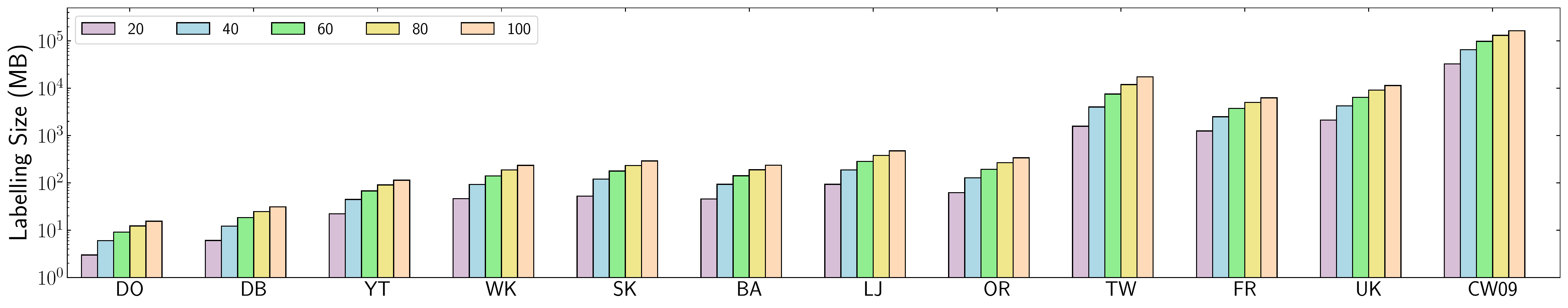}\vspace{-0.4cm}
    \caption{Labelling sizes using QbS under 20-100 landmarks on all the datasets.}
    \label{fig:label-size}\vspace{-0.1cm}
\end{figure*}

Table \ref{tab:compare_construction_time} presents the comparison results of our method with the baselines in terms of query time. Compared with the search-based method Bi-BFS, 
our method QbS can answer queries much more efficiently, i.e., 10-300 times faster than Bi-BFS. 
Particularly, QbS is able to answer queries within milliseconds for 8 out of 12 datasets, and less than 0.5 seconds for the other datasets which have up to 1.7 billion vertices and 7.8 billion edges. We notice that, Twitter has significantly higher query time than Friendster and uk2007. This is because, compared with the other graphs, Twitter has larger shortest path graphs as shown by $size(\Delta)$ in Table \ref{tab:label} due to densely connected vertices with very high degrees.  
For labelled-based methods, the query times of both PPL and ParentPPL are much faster than Bi-BFS. 
However, neither PPL nor ParentPPL is scalable. PPL can only answer queries for the first 5 datasets, while ParenetPPL can only answer queries for the first 2 datasets which have less than 1 million vertices. This is because that constructing labelling information required by these methods is computationally expensive for very large graphs.

\subsection{Effects of Sketching}
We conducted an experiment to understand how sketching improves the performance of query answering in our method.

Figure \ref{fig:coverage} presents the pair coverage ratios of our method QbS using 20-100 landmarks. Here, pair coverage ratio refers to the proportion of queries in which the shortest paths between two vertices go through at least one landmark, among 10,000 queries used in our experiments. We distinguish two cases: \rv{(i) Queries in which \emph{all shortest paths} between two vertices go through at least one landmark; (ii) Queries in which \emph{some but not all shortest paths} between two vertices go through at least one landmark. Pair coverage ratios reflect the effectiveness of sketching used in our method QbS since a sketch cannot guide queries in which none of shortest paths between two vertices go through landmarks.}


From Figure \ref{fig:coverage}, we can see that: (1) When the number of landmarks increases, the pair coverage ratios go up for both Case (i) and Case (ii); nonetheless, the increasing rate generally slows down. (2) For datasets in which graphs have high degree vertices compared with their other vertices, such as Youtube, WikiTalk, Baidu, Twitter, and ClueWeb09, their pair \rv{coverage} ratios are generally higher than the other datasets. This is because these high degree vertices are more likely on the shortest paths of the other vertices. For Friendster, as it does not have   high degree vertices, the pair coverage ratios are quite low. (3) For datasets in which graphs are sparse after removing landmarks that are vertices of high degrees, such as Youtube, WikiTalk, Baidu and ClueWeb09, the percentage of pair coverage ratio for Case (i) among pair coverage ratios for both cases is higher than the other datasets. In Friendster, the degrees of vertices are more evenly distributed; hence, landmarks hardly capture all shortest paths between two vertices and the pair coverage ratios for Case (i) are extremely low. \rv{However, the reasons why query time on Friendster is still fast are twofold: (1) QbS does not store parent information for reverse search since most parent vertices do not lead to shortest paths being recovered, and (2) QbS uses sketches to guide which side to expand for bi-directional searches.}



\subsection{Performance with Varying Landmarks}
We also conducted experiments to evaluate how the number of landmarks may affect the performance of our method. 

\begin{figure*}
    \centering
    \includegraphics[width=16.5cm]{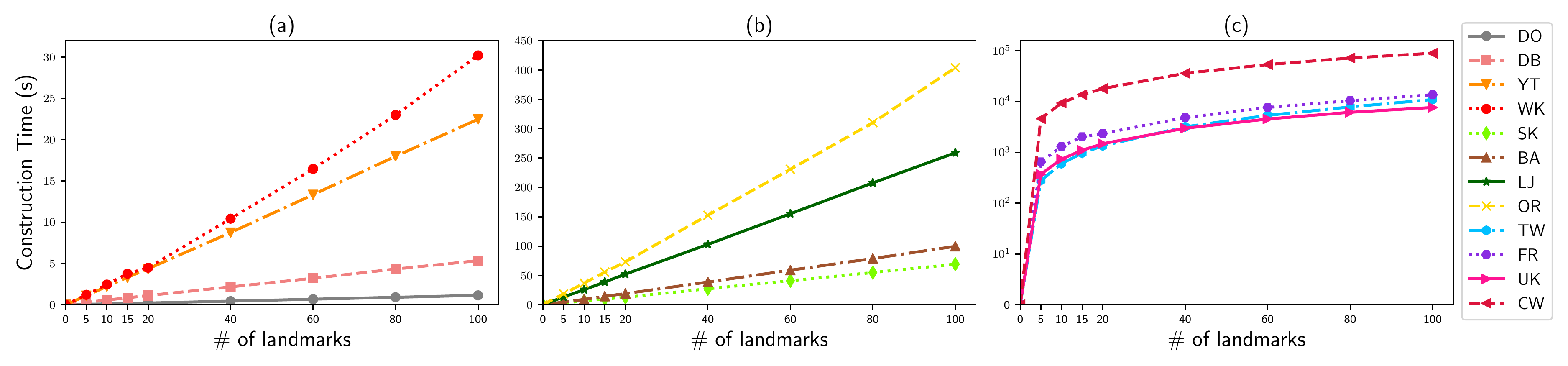}\vspace{-0.3cm}
    \caption{\rv{Construction times using QbS under 0-100 landmarks on all the datasets.}}
    \label{fig:construction-time}
    \centering
    \includegraphics[width=16.5cm]{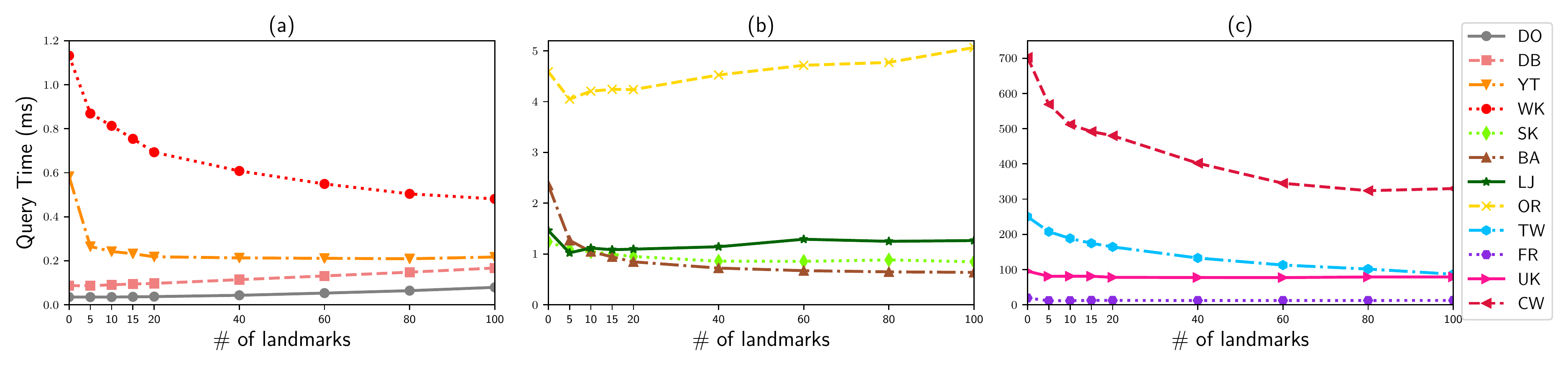}\vspace{-0.3cm}
    \caption{\rv{Average query times using our method QbS under 0-100 landmarks on all the datasets.}}\vspace{-0.3cm}
    \label{fig:query-time}
\end{figure*}

\vspace{-0.1cm}
\subsubsection{Construction Time}
The construction times of our method QbS against different numbers of landmarks (from 20 to 100) are shown in Figure \ref{fig:construction-time}. Generally, the construction time grows linearly. 
In Figure \ref{fig:construction-time}(a)-(b), for datasets with millions of edges, QbS can construct labels under 100 landmarks within at most a few minutes. In Figure \ref{fig:construction-time} (c), for datasets with billions \rv{of} edges, QbS can construct labels within a few hours. It can be seen that the construction time is almost linear in the number of landmarks on each dataset, which confirms the scalability of QbS.

\vspace{-0.1cm}
\subsubsection{Labelling Size}
We compared the labelling sizes of QbS against different numbers of landmarks in Figure \ref{fig:label-size}. For a labelling scheme $\mathcal{L}=(M,L)$, we use $|R|$*8 bits to store labels of each vertex. For $M$, as discussed in Section \ref{subsubsec:labellingsize-20}, the labelling size of a meta-graph is very small, compared with the labelling size of $\Delta$ and $L$. It increases when the number of landmarks becomes larger. Nonetheless, even when $|R|$=100, the labelling size of a meta-graph would still be smaller than 0.01MB.
For $\Delta$, since we store the shortest paths between $|R|^2$ pairs, it grows fast when the number of landmarks increases. 
However, compared with the size of labels in $L$ as shown in Table \ref{tab:label}, $\Delta$ is small. The sizes of shortest paths between vertices with lower degrees are smaller than the ones between vertices with higher degrees. Thus, the labelling size of $\Delta$ does not increase quadratically in the number of landmarks. The sizes of path labelling $L$ are linear in terms of the number of landmarks. 

\vspace{-0.1cm}
\subsubsection{Query Time}
\rv{The impact of varying landmarks on query time is shown in Figure \ref{fig:query-time}.} When the number of landmarks increases, \rv{there are generally three cases: 1) the query times increase, e.g., Douban, DBLP and Orkut; 2) the query times decrease, e.g., 
WikiTalk, Twitter and ClueWeb09; 3) the query times have no significant changes, e.g., LiveJournal and uk2007. If a graph has very high degree vertices, selecting more landmarks often decreases query times because removing more landmarks can further sparsify the graph significantly.} For example, in Twitter, 38 million edges are incident to 20 landmarks, while 100 landmarks have around 123 million edges; accordingly, the query time under 100 landmarks is half as the query time under 20 landmarks. 
\rv{If degrees of vertices in a graph are evenly distributed such as Orkut, more landmarks do not necessarily improve query time; instead, due to increased computational cost for computing a sketch, query time often increases.} 

\vspace*{-0.1cm}
\subsection{Remarks}
\rv{In general, QbS has three sources of efficiency gains when answering shortest-path-graph queries: (1) QbS enables queries to traverse on a graph whose parts with high centrality are sparsified. Thus, although removing a small number of landmarks alone does not significantly reduce the number of edges in a whole graph (e.g., 3.2\% of edges are removed with 20 landmarks in Twitter), the number of edges traversed by queries is significantly reduced (e.g., around 30\% less of edges being traversed by queries in QbS against Bi-BFS). 
(2) QbS uses a sketch to guide the search for each query, further reducing the number of edges being traversed. Take Twitter for example, after adding the guide of sketches on a sparsified graph, 66\% less of edges are traversed in QbS against Bi-BFS. (3) QbS can avoid the computation of shortest paths between high-degree landmarks when two or more landmarks appear on one shortest path, since these shortest paths can be precomputed as discussed in Section \ref{subsec:complexity}. In our experiments, the performance of QbS varies in datasets, depending on how the characteristics of datasets support these sources of gains to speed up query efficiency. }

%% file: section_related-work.tex
\vspace{-0cm}\section{Related Work}
\label{sec:related-work}

\noindent\textbf{Exact algorithms.}
One of the most classical methods for shortest path computation is Dijkstra's algorithm \cite{dijkstra1959note}. It computes a single-source shortest path tree on a weighted graph in time complexity $O(|E|+|V|log|V|)$. For unweighted graphs, breadth-first search (BFS) computes a single-source shortest path tree in $O(|E|)$. However, these methods are very inefficient on large networks. A simple strategy for reducing search space is to employ bi-directional BFS which performs two searches from two given vertices, respectively, based on certain heuristic assumptions \cite{goldberg2005computing, jin2013hub}. 
To further accelerate shortest path computation, a number of methods have been proposed to pre-compute a labelling so as to answer point-to-point shortest path queries online in a shorter time \cite{goldberg2005computing,goldberg2006reach,bast2007transit,goldberg2007point,wagner2007speed,abraham2010highway,wu2012shortest,sankaranarayanan2009path,sanders2005highway, xiao2009efficiently,wei2010tedi}. 
For example, Xiao et al. \cite{xiao2009efficiently} exploited graph symmetry to label shortest paths. 
Though the size of labels has been compressed depending on the symmetric property, the space cost is still high. Later, Wei \cite{wei2010tedi} introduced a method based on tree decomposition for point-to-point shortest path queries. However, most of complex networks have a large component in which vertices are densely connected, making it hard to be decomposed into tree-like structures. 
Several methods have been proposed for finding shortest path distances on complex networks (e.g., \cite{akiba2013fast, fu2013label, akiba2012exploit, hayashi2016fully, farhan2019highly}). Some of them considered answering point-to-point shortest path queries as an extension of answering distance queries, although they did not provide any experiments. For example, Akiba et al. \cite{akiba2013fast} proposed pruned landmark labelling (PLL) which constructs a 2-hop labelling for distance queries by conducting pruned BFSs. Fu et al. \cite{fu2013label} proposed IS-label, a labelling for distance queries on weighted graphs based on an independent set of vertices. Both of these methods discussed labellings for point-to-point shortest path queries by extending labellings for distance queries with parent information, which however require a high space overhead and do not scale to large graphs. In this work, we study the \rv{shortest-path-graph} problem, which is computationally more difficult than the point-to-point shortest path problem, and little attention has previously been given. Our method pre-computes a small-sized distance labelling and can handle complex networks with up to billions of vertices.

\vspace{0.1cm}
\noindent \textbf{Approximate algorithms.} Due to the high computational costs of computing shortest paths, a number of approximate methods for point-to-point shortest path queries have been proposed in the past, including landmark-based methods with acceptable accuracy \cite{gubichev2010fast,zhao2011efficient,tretyakov2011fast}. Specifically, Gubichev et al. \cite{gubichev2010fast} proposed to pre-compute shortest paths from each vertex to each landmark, and then concatenate shortest paths from two vertices to the same landmarks to approximate shortest paths. They also proposed cycle elimination and tree-based sketch to boost  accuracy. 
Zhao et al. \cite{zhao2011efficient} proposed a method, called Rigel, to estimate shortest path distances. They also extended Rigel for approximating shortest paths. 
Tretyakov et al. \cite{tretyakov2011fast} used shortest path trees rooted at landmarks to approximate shortest path distances and search for one shortest path. 
Unlike these approximate algorithms, our work here aims to develop an exact method \rv{to accurately compute a shortest path graph that contains all shortest paths between two given vertices.}